\documentclass[conference]{IEEEtran}
\usepackage{cite}
\usepackage{amsmath,amssymb,amsfonts}
\usepackage{algorithmic}
\usepackage{graphicx}
\usepackage{textcomp}
\usepackage{xcolor}
\def\BibTeX{{\rm B\kern-.05em{\sc i\kern-.025em b}\kern-.08em
    T\kern-.1667em\lower.7ex\hbox{E}\kern-.125emX}}

\usepackage{bm}
\usepackage{amsthm}

\usepackage{multicol}
\newtheorem{theorem}{Theorem}
\newtheorem{definition}{Definition}

\newtheorem{remark}{Remark}
\newtheorem{optimization}{Optimization Problem}
\def\BibTeX{{\rm B\kern-.05em{\sc i\kern-.025em b}\kern-.08em
    T\kern-.1667em\lower.7ex\hbox{E}\kern-.125emX}}

\newcommand{\AvGas}{\bm{g}_{av}}
\newcommand{\PeakGas}{\bm{g}_{peak}}
\newcommand{\PowAl}{\bm{g}}
\newcommand{\AvDemand}{\bm{d}_{av}}
\newcommand{\AvLoss}{\bm{l}_{av}}
\newcommand{\AvWind}{\bm{w}_{av}}
\newcommand{\AvExDem}{\bm{r}_{av}}
\newcommand{\relu}[1]{\left[#1\right]^+}
\usepackage{array}
\makeatletter
\let\old@ps@headings\ps@headings
\let\old@ps@IEEEtitlepagestyle\ps@IEEEtitlepagestyle
\def\psccfooter#1{%
    \def\ps@headings{%
        \old@ps@headings%
        \def\@oddfoot{\strut\hfill#1\hfill\strut}%
        \def\@evenfoot{\strut\hfill#1\hfill\strut}%
    }%
    \def\ps@IEEEtitlepagestyle{%
        \old@ps@IEEEtitlepagestyle%
        \def\@oddfoot{\strut\hfill#1\hfill\strut}%
        \def\@evenfoot{\strut\hfill#1\hfill\strut}%
    }%
    \ps@headings%
}
\makeatother
    
\begin{document}

\title{Effective Capacity of  a Battery Energy Storage System Captive to a Wind Farm\\
}

\author{\IEEEauthorblockN{$^1$Vinay A. Vaishampayan, \textit{Fellow, IEEE}, $^1$Thilaharani Antony, and  $^2$Amirthagunaraj Yogarathnam, \textit{Member, IEEE}}
\IEEEauthorblockA{ 
\textit{$^1$Faculty of the Engineering and Environmental Science}, College of Staten Island-City University of New York, NY, USA\\
\textit{$^2$Interdisciplinary Science Department}, Brookhaven National Laboratory, Upton, NY, USA\\
emails: vinay.vaishampayan@csi.cuny.edu, thilaharani@hotmail.com,  ayogarath@bnl.gov}
%\{email comes here}
}
\maketitle

\begin{abstract}
Wind energy's role in the global electric grid is set to expand significantly. New York State alone anticipates offshore wind farms (WFs) contributing 9GW by 2035. Integration of energy storage emerges as crucial for this advancement. In this study, we focus on a WF paired with a captive battery energy storage system (BESS). We aim to ascertain the capacity credit for a BESS with specified energy and power ratings. Unlike prior methods rooted in reliability theory, we define a power alignment function, which leads to a straightforward definition of capacity and incremental capacity for the BESS. We develop a solution method based on a linear programming formulation. Our analysis utilizes wind data, collected by  NYSERDA off Long Island's coast and load demand data from NYISO. Additionally, we present theoretical insights into BESS sizing and a key time-series property influencing BESS capacity, aiding in simulating wind and demand for estimating BESS energy requirements.

%Wind energy contributions to the electric grid are poised for significant growth, globally. Contributions of 9GW from offshore wind farms (WFs) are expected in New York State alone by 2035. Energy storage emerges as a pivotal component facilitating this integration. To that end,  in this work, we consider a WF and a captive battery energy storage system (BESS). Our objective is to determine the power capacity credit for a BESS with a given nameplate energy (MWh) and power capacity (MW). We adopt a time-series approach, unlike previous approaches to the problem, which are based on reliability theory. This approach leads to a natural and simple definition of capacity and incremental capacity for the BESS. A solution method, based on a linear programming formulation is developed. Results are presented for wind data collected by NYSERDA off the Long Island coast in New York State, and demand data obtained by NYISO. Finally, theoretical results about the BESS size are presented, along with a characterization of an important time series property that determines the capacity of the BESS. This is useful in simulating wind and demand for estimating BESS energy requirements.
\end{abstract}


\section{Introduction}
The assessment of a power source's capacity credit has traditionally been rooted in thermal power and relies on probabilistic methods, typically involving the calculation of metrics like the Loss of Load Probability (LOLP) and related quantities such as the effective load carrying capability (ELCC). Over time, these methodologies have been adapted to account for renewable resources such as wind farms (WFs) and solar-PV~\cite{Garver:1966, KMD:2011}, as well as energy-limited sources like battery energy storage systems (BESSs) \cite{Sio:2014, Parks:2019}.  The challenge of sizing BESSs to mitigate the variability of power sources or to provide peaking capacity has garnered significant attention in recent research. %Various studies have explored different aspects of this challenge, including optimal sizing methodologies, performance assessment, and the impact of storage system deployment on grid stability and reliability. 
 A multi-objective optimization model which considers economic and reliability aspects to determine optimal sizing of BESSs is proposed in~\cite{li2019multi}. A stochastic optimization framework for  sizing BESSs which considers uncertainties in renewable energy generation is developed in~\cite{wang2020stochastic}. 
Studies such as~\cite{tuohy2009impact,zhang2018performance} have evaluated the effectiveness of BESSs in smoothing renewable energy output fluctuations and reducing the need for conventional generation reserves. %The impact of different energy storage system configurations on grid stability and reliability metrics is studied in~\cite{chen2021impact}.

The deployment of BESSs for providing peaking capacity has also been a subject of interest. In~\cite{wu2017study,wang2019optimal} the role of BESSs in meeting peak demand and reducing the reliance on expensive peaking power plants is considered. These studies demonstrate the potential of BESSs to provide flexible and reliable peaking capacity, particularly when coupled with renewable energy sources. The complexity of computing the effective capacity motivated the simplified mathematical model presented in~\cite{Dejvises:2016} and approximations in \cite{BESS:Denholm2020}, notably the peak demand reduction credit. Our work is in the spirit of~\cite{Dejvises:2016,BESS:Denholm2020}---to present and work with a simple analytic model. Specifically, our contribution is the definition of the power alignment function and its computation, in showing that effective capacity and incremental capacity values can be derived from it, and in obtaining theoretical results  and statistical insights into wind and demand modeling.   
%Research efforts documented in \cite{Dejvises:2016, BESS:Mak2010, BESS:Denholm2020} address the complexities of determining the optimal size of energy storage systems. One prominent approach in this domain is the peak demand reduction credit (PDRC)~\cite{BESS:Denholm2020}, which focuses on mitigating peak demand through the deployment of storage units. The PDRC is evaluated by assessing various storage power capacities and determining the corresponding energy capacity required to reduce peak demand. Notably, increasing power capacity while holding energy capacity constant results in diminishing PDRC.

%These studies often rely on simulation-based analyses to explore the behavior of storage systems, particularly around peak load times. The insights gained from such analyses provide valuable intuition for optimizing energy storage system configurations and enhancing their effectiveness in managing power source variability.

The key role of a BESS is in time-shifting or aligning the power availability  to demand. Since a BESS does not generate power, it \emph{must} be coupled with a power source to assess its capacity credit. In this context, we aim to define the capacity credit that a BESS should receive when it is tied in to a time-varying power source such as a wind farm. This is referred to as a  captive BESS for that wind farm. Our approach leads to a natural definition of power capacity and incremental power capacity for a captive BESS, based on a time series methodology rather than traditional LOLP-based methods.
The paper is organized as follows. Section~\ref{sec:Problemformulation} defines  the power alignment function, which is fundamental to capacity estimation. Section~\ref{sec:batterysize} describes the physical setup, the governing equations and the method for computing the power alignment function and the capacity.  Numerical capacity estimates based on NYSERDA wind speed measurements and NYISO demand data are in Sec.~\ref{sec:results}. Theoretical results are sketched in Sec.~\ref{sec:theory}. The paper is summarized and conclusions are presented in Sec.~\ref{sec:summary}. Some proofs are presented in the Appendix.

\section{The Power Alignment Function} \label{sec:Problemformulation}
\begin{figure}[htbp] %  figure placement: here, top, bottom, or page
   \centering
   \includegraphics[width=3in]{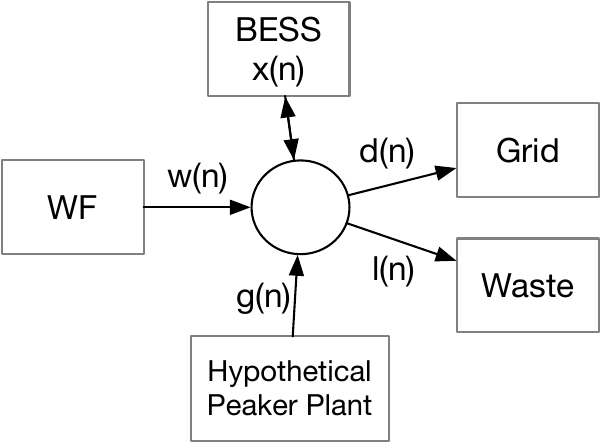} 
   \caption{Test setup for calculating the power alignment function $g(B,P)$.}
   \label{fig:testsetup}
\end{figure}
\begin{figure}[htbp] %  figure placement: here, top, bottom, or page
   \centering
    \vspace{-15pt}
   \includegraphics[width=2.55in]{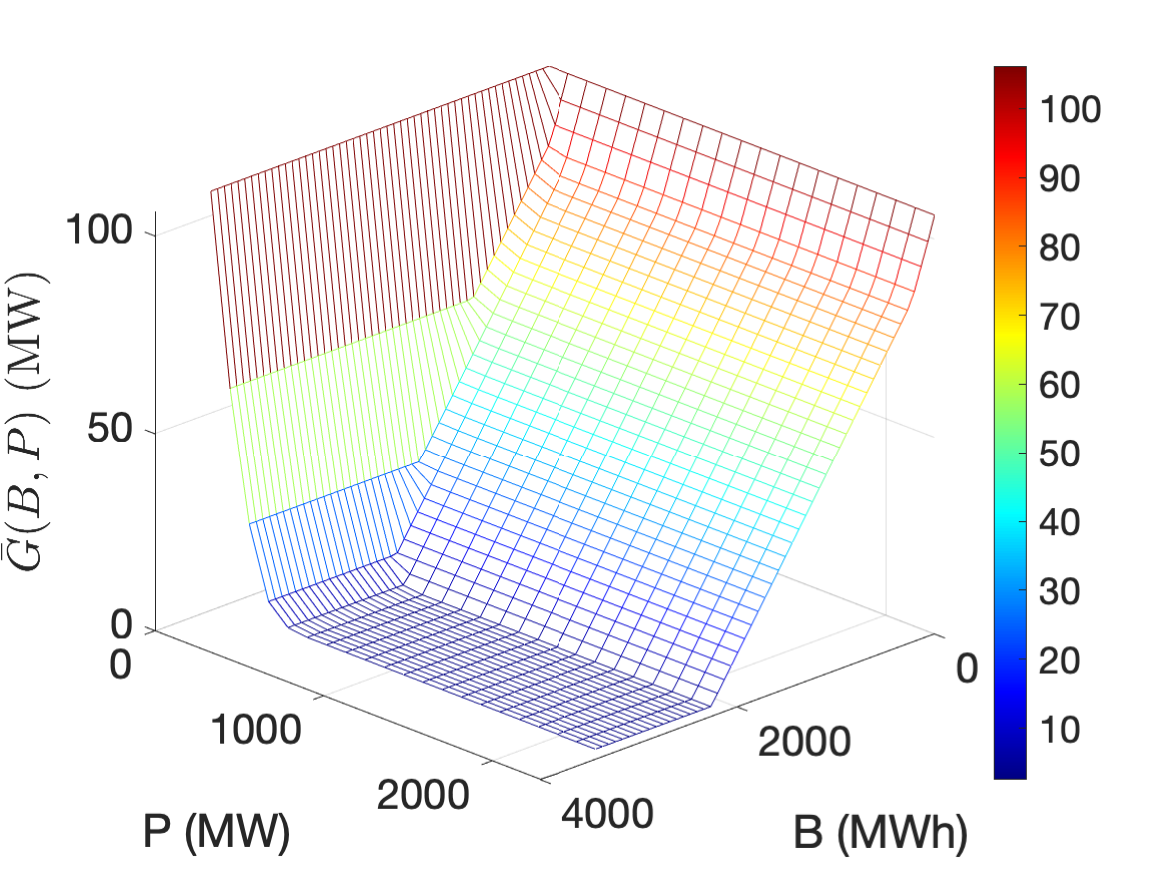}  
    \vspace{-7pt}
   \includegraphics[width=2.55in]{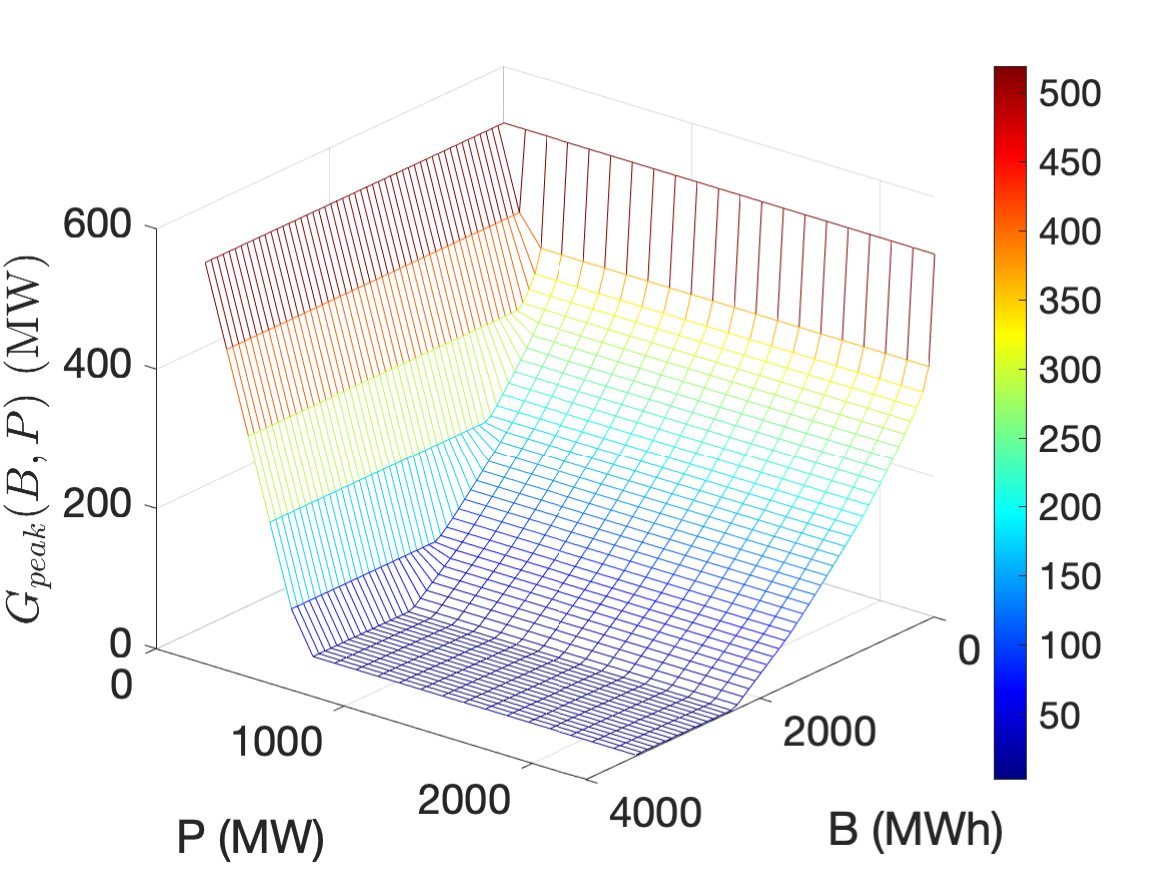} 
\vspace{10pt}
   \caption{Power Alignment Functions for a Day 1, $\AvGas(B,P)$   (top), $\PeakGas(B,P)$  (bottom). The corresponding wind power and demand traces are in Fig.~\ref{fig:threedaysdata}.}
   \vspace{-10pt}
   \label{fig:gbarbp}
\end{figure}
In this section, we describe our framework  for assigning credit to a captive BESS, and define some basic functions.
Our setup consists of  a WF, a captive BESS, a load, and a hypothetical peaker power plant as illustrated in Fig.~\ref{fig:testsetup}.  Our approach is to define the \emph{Power Alignment Function}, $\PowAl(B,P)$ which measures the power required from the peaker plant as a function of the energy rating, $B$ MWh, and power rating, $P$ MW, of the BESS.  From this we derive two capacities, the absolute capacity (hereafter referred to as the  capacity) and the incremental capacity. The basic idea is as follows. A mismatch between demand and supply results in incomplete utilization of wind power. The capacity credit for a BESS  is the amount by which this utilization can be increased. We measure this gain in utilization by the reduction in power required from a hypothetical peaker power plant. Thus the fundamental quantity is the amount of peaker plant power for a BESS system with a given $(energy,power)$ rating $(B,P)$, which is the above-mentioned power alignment function. Precise definitions follow.%$\AvGas(B,P)$ and $\PeakGas(B,P)$ 
\begin{definition} 
\label{def:PA}
The power alignment function for a given variable energy source, a given demand and a BESS having energy rating $B$ MWh and power rating $P$ MW is the \emph{minimal} peaker plant power, $\PowAl(B,P)$,  required such that demand can be met at all times. 
\end{definition}
%which measures the  peaker plant power required for a WF with a captive BESS  and with a 
\begin{remark}The peaker plant power could be either the average power or the peak power---the specific measure used will be identified with a suitable subscript. Thus $\AvGas(0,0)$ and $\PeakGas(0,0)$ measure the average and peak peaker plant power, respectively, when there is no captive BESS. 
\end{remark}
\begin{remark}
While the definition places no restrictions on the average demand and wind power, a meaningful tradeoff between BESS rating and peaker plant power will usually be obtained when  the average wind  and demand powers are equal.  
\end{remark}
\begin{definition}
The capacity assigned to the BESS with energy rating $B$ and power rating $P$ is  $$\kappa(B,P)=\PowAl(0,0)-\PowAl(B,P),$$ which is the reduction in  peaker plant power due to the BESS. The normalized capacity is given by $\kappa(B,P)/\PowAl(0,0)$. 
\end{definition}
\begin{definition}
The \emph{incremental}  $P$-capacity for  power rating $P$ and energy rating $B$ with $B=f(P)$, where $f$ is a function that defines the dependence of $B$ on $P$  is given by $-d g(f(P),P)/dP$, which can be obtained in  terms of the gradient of the power alignment function $$\nabla  \PowAl (B,P)=\begin{pmatrix} \partial \PowAl(B,P)/\partial B,  & \partial \PowAl(B,P)/\partial P \end{pmatrix}$$ using 
$$
-\nabla \PowAl(f(P),P)\cdot\mathbf{v}
$$
with $\mathbf{v}=\begin{pmatrix} df(P)/dP, & 1 \end{pmatrix}$. 

On the other hand if $P=f(B)$ defines the dependence of $P$ on $B$, then the incremental $B$-capacity is $-d g(B,f(B))/dB$ which can be obtained in terms of the gradient of the power alignment function using
$$
-\nabla \PowAl(B,f(B))\cdot\mathbf{v}
$$
where $\mathbf{v}=\begin{pmatrix} 1, & df(B)/dB\end{pmatrix}$. 
\end{definition}
\begin{remark}
The operator ''$\cdot$'' is the dot product, also known as the inner product, between two vectors. 
As an example, if the power rating $P$ is held constant and the energy rating $B$ is increased, then $\mathbf{v}=\begin{pmatrix} 1, & 0 \end{pmatrix}$ and the incremental $B$-capacity is simply $-\partial \PowAl(B,P)/\partial B$ and the incremental $P$-capacity is undefined. As another example if $B=4P$, corresponding to $4$-hour storage, then the incremental $P$-capacity is obtained by setting
$\mathbf{v}=\begin{pmatrix} 4, & 1 \end{pmatrix}$ and is given by $-(4\partial \PowAl(4P,P)/\partial B + \partial \PowAl(4P,P)/\partial P)$. On the other hand the incremental $B$-capacity for $4$-hour storage is obtained by writing $P=B/4$, setting $\mathbf{v}=\begin{pmatrix} 1, & 1/4 \end{pmatrix}$, and is given by
$-(\partial \PowAl(4P,P)/\partial B + (1/4)\partial \PowAl(4P,P)/\partial P)$.
\end{remark}
The incremental capacity captures the incremental power savings when additional BESS storage is added. The method for computing $\PowAl(B,P)$ is presented in Sec.~\ref{sec:batterysize}. 
%the negative of the partial derivative $-\partial \AvGas(B,P)/\partial B$ ($-\partial {G}_{peak}(B,P)/\partial B$). Other incremental capacities can be similarly defined, e.g., the incremental (with respect to power) average  capacity for a $4$ hour BESS is the derivative $-d\AvGas(4P,P)/dP$. 

Typical examples of the power alignment function that we have computed are shown in Fig.~\ref{fig:gbarbp}. These are based on NYSERDA wind speed measurements~\cite{b1} and NYISO demand traces~\cite{b2} shown in Fig.~\ref{fig:threedaysdata}. The surface is seen to be constant along the boundary $B=0$ and $P=0$ and decreases as $B$ and $P$ increase until it flattens out.

It is worth noting that our definition of incremental capacity is far simpler and easier to calculate than the traditional estimate, which involves increasing the generation capacity and then incrementing the demand to maintain a constant LOLP. Our approach also has the benefit that it is amenable to theoretical analysis. Specifically, the value $\AvGas(B,0)=\AvGas(0,P)$,  the value of $B$, $P$ required to reduce $\AvGas(B,P)$ to zero (under some conditions) and a key characteristic of the excess demand time series that determines the shape of $\AvGas(B,P)$ are identified in Sec.~\ref{sec:theory}.

\section{Computing the Effective Capacity}
\label{sec:batterysize}
In this section we describe the physical setup, mathematical notation, governing equations for the BESS state and formulate the solution as an optimization problem.

\subsection{Definitions of Energy and Power Terms}
Energy is measured in MWh, and power in MW, and this is assumed in the equations that follow. Here we develop notation and equations of state for the BESS. Our system consists of an integrated WF and battery storage system, which feeds/absorbs power to/from an electrical grid, as illustrated in Fig.~\ref{fig:testsetup}. A hypothetical peaker plant makes up power shortfalls. 
Electrical energy attributed to wind is represented by the time series $w(n),~n=0,1,2,\ldots$, where  $n$ indexes the sample, corresponding to the time instant $n\Delta$  (measured in hours from a suitable time origin). Thus $w(n)$ is the total wind-generated  energy between time instants $(n-1)\Delta$ and $n\Delta$. The corresponding energy load or demand is denoted by $d(n),~n=0,1,2,\ldots$ while the  peaker plant energy is denoted by $g(n)$. The BESS has energy rating $B$ and power rating $P$. Let $l(n)$ be the wind energy that is lost because supply exceeds demand and the amount of energy that can be stored at time instant $n$. Let $x(n)$ be the battery state at time $n$. We assume that $\alpha$, $0\leq \alpha \leq 1$ is the battery energy loss over duration $\Delta$. 
The average wind, demand, peaker plant and loss powers over a window of $N$ samples are given by $\AvWind=(1/(N\Delta))\sum_{i=1}^Nw(i)$, $\AvDemand=(1/(N\Delta))\sum_{i=1}^Nd(i)$, $\AvGas=(1/(N\Delta))\sum_{i=1}^Ng(i)$ and $\AvLoss=(1/(N\Delta))\sum_{i=1}^Nl(i)$, respectively. The peak peaker plant power is given by $\PeakGas=(1/\Delta)\max_{i\in \{1,2,\ldots,N\}}g(i)$.
 
From the definition of the power alignment function (Definition \ref{def:PA}) it follows that $\PowAl(B,P)$ is the minimum value of $\AvGas$ ($\PeakGas$) for given $w(n),~d(n),~n=1,2,\ldots,N$ and BESS parameters $B$ and $P$, if the measure of interest is the average (peak) peaker plant power.
%$\bar{W}=60W(N)/(N\Delta)$, $\AvgDemand=60D(N)/(N\Delta)$ and the average gas and loss powers are 
% $$\bar{G}=\frac{60}{N\Delta}\sum_{n=1}^N g(n)$$
%and 
%$$\bar{L}=\frac{60}{N\Delta}\sum_{n=1}^Nl(n),$$
%respectively. 
\subsection{Linear Program for the Power Alignment Function}
Calculation of the power alignment function is formulated as a linear program (LP) in this section. Both the average and maximum peaker plant power problems are solved with the formulation below. Our objective is to minimize the peaker plant power for a given sequences of wind energy, demand, and constraints imposed by the BESS specification on total stored energy and the rate at which it can be charged and discharged. Any algorithm or protocol that dictates the evolution of the BESS state must obey the energy conservation law 
\begin{eqnarray}
x(n)  =  \alpha x(n-1)+(w(n)-d(n))+(g(n)-l(n)), & \  &\nonumber \\
 ~n=1,2,\ldots,N.
\label{eqn:newcompactevo}
\end{eqnarray}
This leads to the following problem formulation.

\begin{optimization}
Given  wind and demand energy sequences, $w(n),~d(n),~n=1,2,\ldots,N$, respectively, minimize the  peaker plant power  by selecting the starting state $x(0)$ and non-negative sequences $g(n), l(n),~n=1,2,\ldots,N,$
subject to the stored energy constraints,
\begin{equation}
0 \leq x(n) \leq B,~ n=0,1,2,\ldots,N,
\end{equation}
%\begin{equation}
%x(0)\leq x(n),
%\end{equation}
the power constraints,
\begin{equation}
|x(n)-\alpha x(n-1)| \leq \Delta P,~n=1,2,\ldots,N,
\end{equation}
and the energy conservation constraints on $x(n)$ given by \eqref{eqn:newcompactevo}.
\label{op:op1}
\end{optimization}

%\begin{optimization}
%Given  wind and demand energy sequences, $w(n),~d(n),~n=1,2,\ldots,N$, respectively, minimize the average peaker plant power  given by $\AvGas(B,P)$ over non-negative sequences $g(n), l(n),~n=1,2,\ldots,N,$ and starting state $x(0)$,
%subject to the constraints
%\begin{equation}
%0 \leq x(n) \leq B,~ n=0,1,2,\ldots,N,
%\end{equation}
%\begin{equation}
%x(0)\leq x(n),
%\end{equation}
%\begin{equation}
%|x(n)-\alpha x(n-1)| \leq \Delta P,
%\end{equation}
%where $x(n)$ obeys the recursion \eqref{eqn:newcompactevo}.
%\label{op:op1}
%\end{optimization}
%
%\begin{optimization}
%Given wind and demand energy sequences, $w(n),~d(n),~n=1,2,\ldots,N$, respectively, minimize the peak peaker plant power  given by $\PeakGas(B,P)$ over non-negative sequences $g(n), l(n),~n=1,2,\ldots,N,$ and starting state $x(0)$,
%subject to the constraints
%\begin{equation}
%0 \leq x(n) \leq B,~ n=0,1,2,\ldots,N,
%\end{equation}
%\begin{equation}
%x(0)\leq x(n),
%\end{equation}
%\begin{equation}
%|x(n)-\alpha x(n-1)| \leq \Delta P,
%\end{equation}
%where $x(n)$ obeys the recursion \eqref{eqn:newcompactevo}.
%\label{op:op2}
%\end{optimization}
%
%The solutions to  Optimization Problems \ref{op:op1} and \ref{op:op2}  are readily obtained using standard linear programming solvers. 

%We emphasize that the LP formulation obtains the minimal peaker plant power \emph{and} the initial battery state in a single shot. 

\subsection{Greedy Charging Protocol}
The greedy charging protocol is a specific (and obvious) battery charging protocol. It is described in detail in this section.
%The initial state of the battery $x(0)$ satisfies $0 \leq x(0) \leq B$. 

Let
\begin{equation}
f(n)  =  \alpha x(n-1) + (w(n)-d(n)),~n=1,2,\ldots .
\label{eqn:recursion1}
\end{equation}
$f(n)$ is an intermediate quantity used to simplify the description of the battery state evolution given by
%\end{multicols}
%\par\noindent\rule{\dimexpr(0.5\textwidth-0.5\columnsep-0.4pt)}{0.4pt}%
%\rule{0.4pt}{6pt}
\begin{eqnarray}
\lefteqn{x(n)  =} & \  \nonumber  \\
&  \left\{ \begin{array}{ll} B, & f(n) > B, \\
        \Delta P +\alpha x(n-1), & \Delta P+\alpha x(n-1) < f(n) \leq B, \\
         f(n), & max\{-\Delta P + \alpha x(n-1), 0 \}  \leq  \\
         \ & \leq f(n) \leq  \\
         \ & \leq \min\{\Delta P+\alpha x(n-1),B\},  \\
         -\Delta P + \alpha x(n-1), & 0< f(n)  \leq -\Delta P +\alpha x(n-1), \\
         0, & f(n) < 0. \end{array}
         \right. \nonumber \\
         & 
         \label{eqn:funditerate}
\end{eqnarray}
%\vspace{\belowdisplayskip}\hfill\rule[-6pt]{0.4pt}{6.4pt}%
%\rule{\dimexpr(0.5\textwidth-0.5\columnsep-1pt)}{0.4pt}
%\begin{multicols}{2}
%It is useful to define the function $\gamma$ and write \eqref{eqn:funditerate} in more compact form as 
%\begin{eqnarray}
%x(n) & = & \gamma(f(n),B_{min},B_{max}) \nonumber \\
%& = & \max\left[\min\left[f(n),B_{max}\right],B_{min}\right].
%\end{eqnarray}
%
%If $f(n)<0$, there is excess demand, which can only be satisfied by using the hypothetical peaker plant.  
The peaker plant energy, $g(n)$, is given by
\begin{equation}
g(n)=\left\{\begin{array}{cc} x(n)-\alpha x(n-1)+  &  \\
+d(n)-w(n), & d(n)-w(n) > 0, \\
0, & d(n)-w(n)\leq 0,
\end{array} \right.
\label{eqn:gas}
\end{equation}
and the lost wind energy, $l(n)$, is given by
\begin{equation}
-l(n)=\left\{\begin{array}{cc} x(n)-\alpha x(n-1) + \  \\ 
+d(n)-w(n), & d(n)-w(n) < 0, \\
0, & d(n)-w(n) \geq 0.
\end{array} \right.
\label{eqn:loss}
\end{equation}
It follows immediately that the greedy protocol respects energy conservation given in \eqref{eqn:newcompactevo}, as  expected. %It is worth noting that \eqref{eqn:newcompactevo} is a statement of the  energy conservation law and must be respected by any battery charging protocol.

The power alignment function $\PowAl(B,P)$ is, in general,  a lower bound for the peaker plant power required by the greedy protocol described in this section. But it is shown to match the lower bound in a specific case in Section~\ref{sec:theory}.
%We start by making the observation that
%\begin{eqnarray}
%x(n)-B_{min}=\gamma(f(n)-B_{min},0,B_{max}-B_{min}), 
%\nonumber \\
%~~~~~~~~~~~~~~~~~~~~~~~~~~~~n=1,2,\ldots.
%\end{eqnarray}
%Thus we may set $B_{min}$ to $0$ and replace $B_{max}$ by $B_{max}-B_{min}$ in equations \eqref{eqn:recursion1}--\eqref{eqn:loss} with no loss in generality. We will do so in the sequel and will  assume that $0\leq x(0) \leq B_{max}$.

%To summarize, the following equations govern the evolution of the battery state.
%\begin{eqnarray}
%f(n) & = & x(n-1)+w(n)-d(n)  \label{eqn:compactevo}\\
%x(n) & = & \max[\min[f(n),B_{max}],0],  \\
%g(n) & = & \min[-f(n),0]  \\
%l(n) & = & \max[f(n)-B_{max},0],
%\label{eqn:fullevo}
%\end{eqnarray}
%for $n=1,2,\ldots,N$, with initial condition $x(0)=B_0$. 
%if sequences $g$ and $l$ are known.

\section{Numerical Results based on Measured Wind Velocity}
\label{sec:results}
\subsection{Wind Power Data}\label{sec:windspeed}
\begin{figure}[htbp] %  figure placement: here, top, bottom, or page
   \centering
   \vspace{-15pt}
   \includegraphics[width=2.75in]{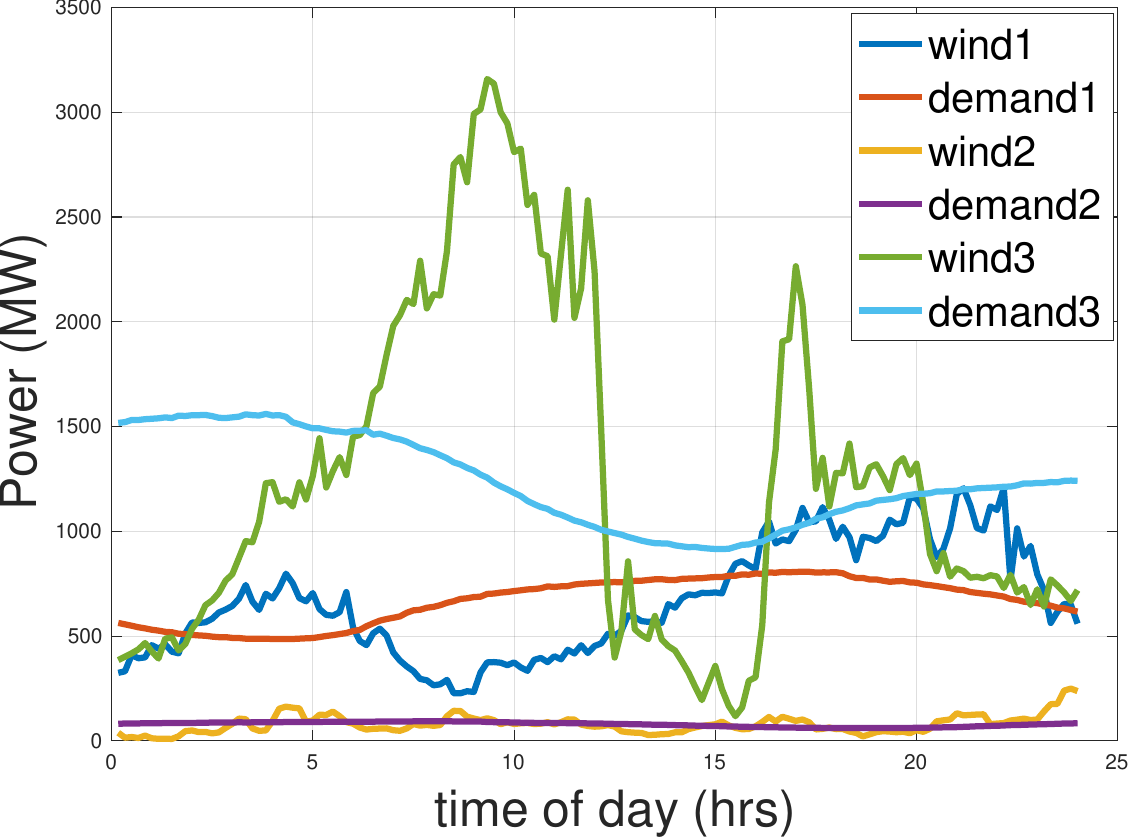} 
   \caption{Wind and demand traces for 9/20/2019 (Day 1), 9/03/2019 (Day 2) and 9/04/2019 (Day 3). Average demand power has been adjusted to equal average wind power.  }
   \label{fig:threedaysdata}
\end{figure}
\begin{figure}[htb] %  figure placement: here, top, bottom, or page
   \centering%   \vspace{-1.6in}
%  \includegraphics[width=0.2\textwidth]{GbarBmaxDay1.pdf}
%%\vspace{-1.6in}
%\includegraphics[width=0.2\textwidth]{GbarBmaxDay2.pdf}
%%\vspace{-1.6in}
%\includegraphics[width=0.2\textwidth]{GbarBmaxDay3.pdf}
%%\vspace{-1.6in}
   \includegraphics[width=0.24\textwidth]{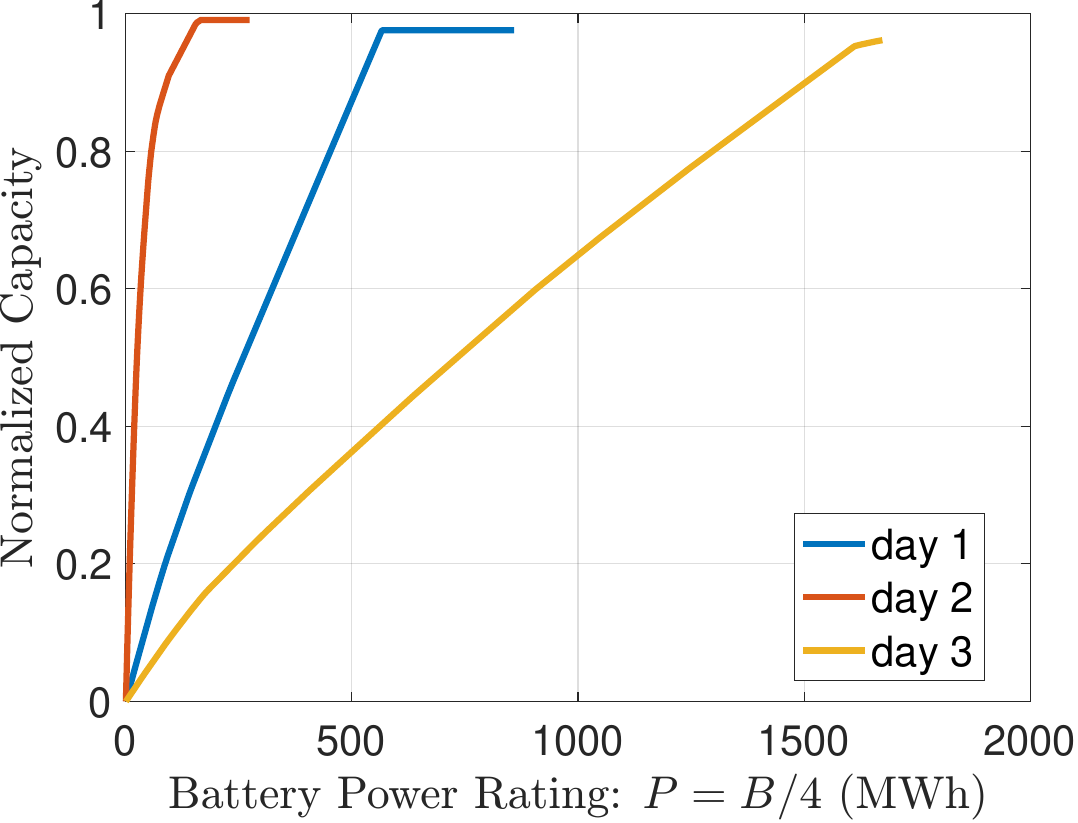} \includegraphics[width=0.24\textwidth]{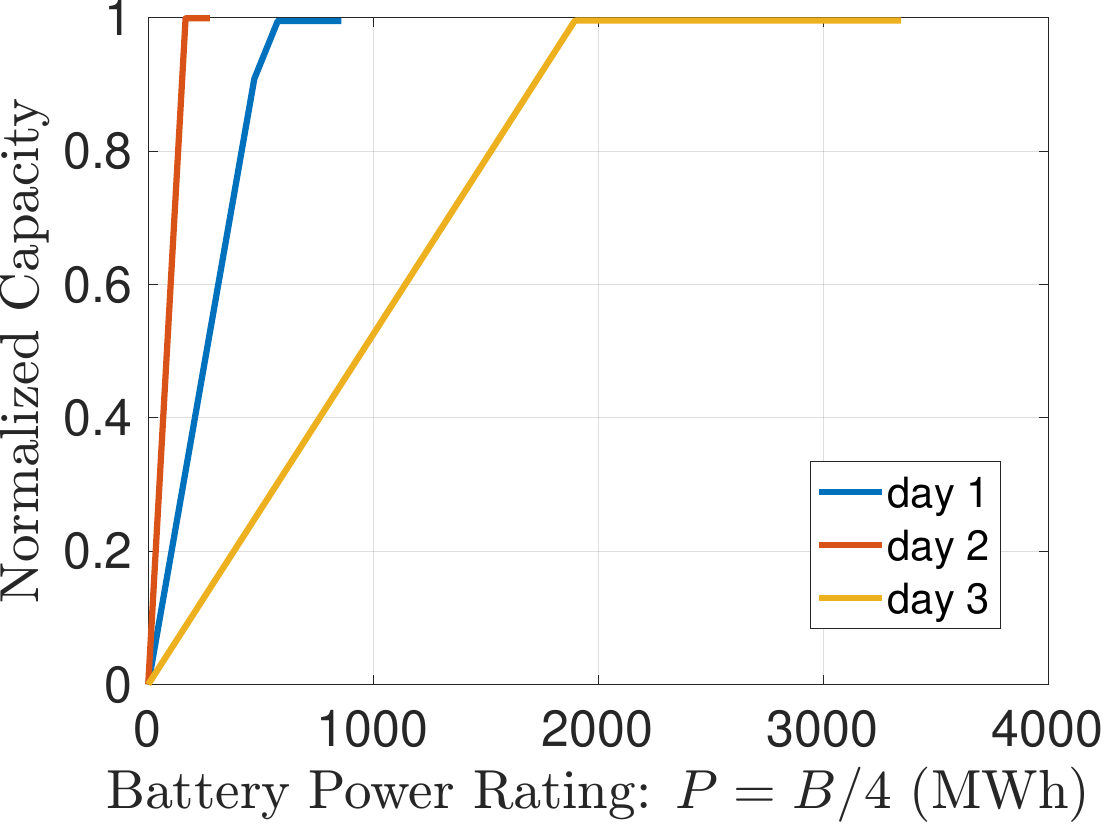}\\
   \includegraphics[width=0.24\textwidth]{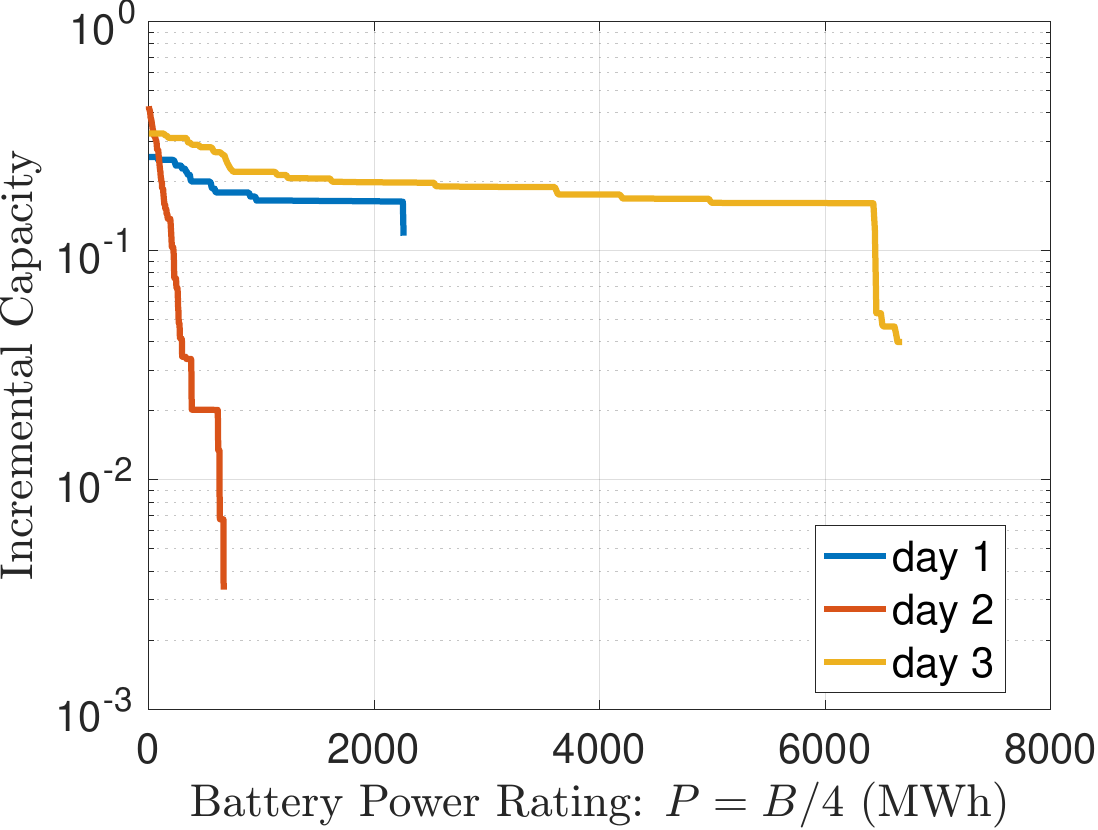}
   \includegraphics[width=0.24\textwidth]{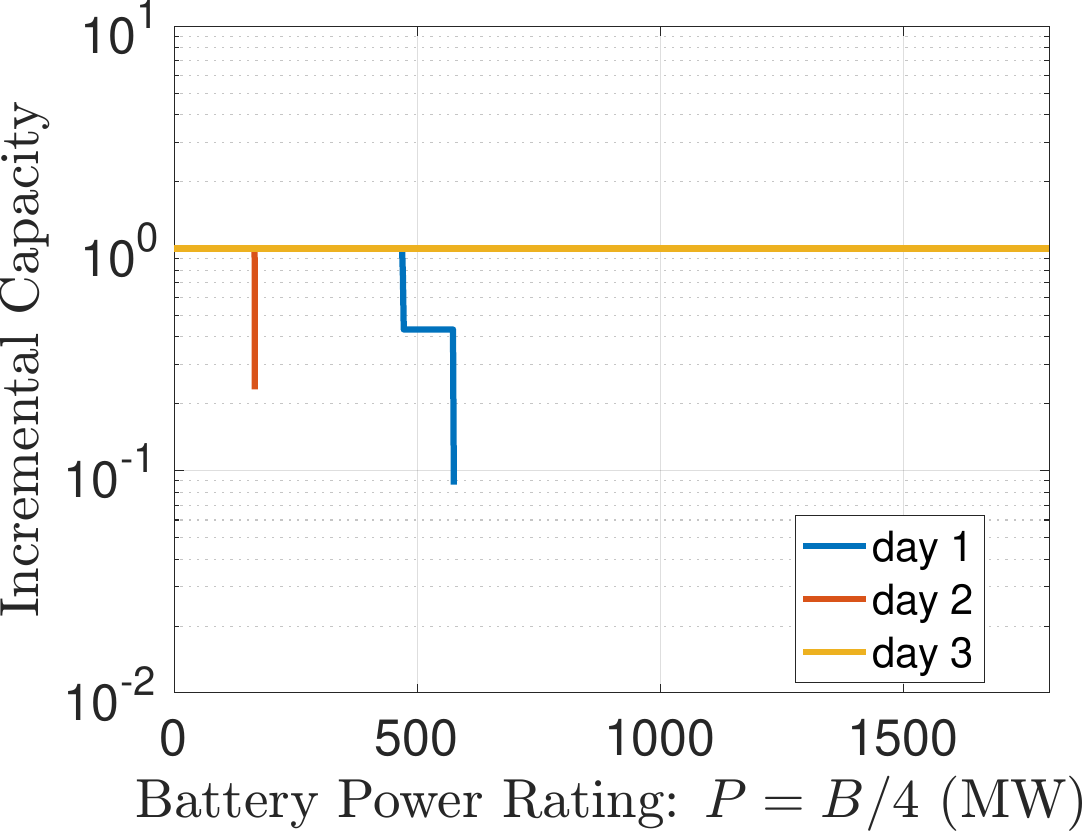}
%\vspace{-0.7in}
    \caption{(Top left) Normalized average capacity $(\AvGas(0,0)-\AvGas(B,B/4))/\AvGas(0,0)$, (Top right) Normalized peak capacity, (Bottom left) Incremental average capacity $-d\AvGas(B,P)/dP$ (for a 4 hour BESS i.e. $B=4P$) (Bottom right) incremental peak capacity.  $5$\% battery loss. Based on wind and demand traces for the 3 days shown in Fig.~\ref{fig:threedaysdata}.}
    \vspace{-10pt}
    \label{fig:tradeoff}
\end{figure}
NYSERDA Hudson North horizontal wind speed data~\cite{b1} was used to estimate wind energy $w(n)$. NYISO data was used for the demand power, after suitable scaling to equalize demand and wind power. The NYSERDA wind speed data was obtained at a sampling rate of $50$ samples per second. Measurements were then averaged over non-overlapping $10$-minute intervals. The accuracy of the wind speed measurement is 0.1 m/s and range of measured velocities  is from $1 m/s$ to $80 m/s$. Wind power in MW for a wind speed $v$ m/s and swept area $A~m^2$ was estimated by $(10^{-3}/2)\rho C_p v^3 A$ for a wind turbine with radius $118$ m and the air density $\rho=1.225~kg/m^3$. We show in Fig.~\ref{fig:threedaysdata} traces of estimated wind power and demand on three different days in 2019,  9/20 (Day 1), 9/03 (Day 2) and 9/04 (Day 3). 
\subsection{Power Alignment Function and Capacity Estimates}
The Power Alignment Function is shown in Fig.~\ref{fig:gbarbp}, and normalized and incremental capacity estimates are shown in Fig.~\ref{fig:tradeoff} for the three days whose time series are presented in  Fig.~\ref{fig:threedaysdata}. The battery energy loss was set to 5\% over a 24 hour period. In all cases, we assume a $4$ hr battery, which means that the power rating $P$ satisfies $P=B/4$. The Normalized Average Capacity for a $4$ hour BESS is given by $(\AvGas(0,0)-\AvGas(B,B/4))/\AvGas(0,0)$. The incremental (w.r.t. power) average capacity is the derivative $-d\AvGas(4P,P)/dP$ for a $4$ hr BESS. Similar quantities can be defined for the peak capacity.

The lost average power, when there is no storage, is  $\AvGas(0,0)=106,~16 \mbox{ and } 332$ MW for Day 1 (as is seen in Fig.~\ref{fig:gbarbp}), Day 2 and Day 3, respectively. From Fig.~\ref{fig:tradeoff} (top), we see that in all cases, the lost power can be recovered almost completely. However, the battery energy capacity varies significantly by day.
The BESS energy rating required to recover 50\% of the lost wind power is $B=1037.0,~98.5 \mbox{ and } 2931.6$ MWh, and the corresponding BESS power rating is $259.3,~24.6 \mbox{ and }732.9$ MW, for Day 1, Day 2, and Day 3 respectively. Thus on Day 1, recovery of  $106$~MW of wind power requires a 4 hr BESS with power rating $259.3$~MW, an efficiency of $40.9\%$. The efficiencies on Day 2 and Day 3 are $64.94 \% \mbox{ and } 45.30 \%$, respectively. In all cases, the incremental capacity is seen to decline or remain flat with $B$, though the rates of decline are significantly different. This depends on the run structure of $R(n)$ in \eqref{eqn:cumexdemand}, as will be seen in Sec.~\ref{sec:theory}.

\section{Theoretical Results}
 \label{sec:theory}
We have presented numerical results which give a general idea about the shape of the power alignment function $\PowAl(B,P)$.  In this section we develop a quantitative characterization of $\PowAl(B,P)$ based on characteristics of the demand sequence $d(n)$ and the wind sequence $w(n)$.  We do so by first developing a general lower bound for $\AvGas(B,P)$  in Sec.~\ref{sec:theory-a}. This bound is refined, based on characteristics of the wind and demand sequences in Sec.~\ref{sec:theory-b}, which leads to a natural graph interpretation of the lower bound as a flow problem.
 \subsection{A General Lower Bound for $\AvGas(B,P)$}
 \label{sec:theory-a}
 For convenience of the reader we restate the energy conservation equations \eqref{eqn:newcompactevo}, and define two new sequences $r(n)$ and $h(n)$ below
\begin{equation}
x(n)=\alpha x(n-1)+\underbrace{(w(n)-d(n))}_{-r(n)}+\underbrace{(g(n)-l(n))}_{h(n)}.
\label{eqn:conservation}
\end{equation}
Note that, upon defining $\relu{x}=\max \{x,0\}$, and using the fact that $g$ and $l$ are non-negative, we get  $g(n)=\relu{h(n)}$ and $l(n)=\relu{h(n)}-h(n)$, i.e, $g$ and $l$ can be recovered from $h$. A similar remark holds for $r$, $w$ and $d$.
 
The energy conservation equations \eqref{eqn:conservation} can be put into matrix form
\begin{equation}
\begin{pmatrix}h(1) \\h(2)\\ \vdots \\h(N) \end{pmatrix}=\underbrace{\begin{pmatrix}-\alpha & 1 & & &  \\ \  & -\alpha & 1 & & \\ &  & & \ddots & \\ & & & -\alpha & 1\end{pmatrix}}_{\mathbf{A}} \begin{pmatrix}x(0) \\x(1)\\ \vdots \\x(N) \end{pmatrix} + \begin{pmatrix}r(1) \\r(2)\\ \vdots \\r(N) \end{pmatrix}
\end{equation}
where $\mathbf{A}$ is an $N\times (N+1)$ matrix. 
Multiplication of both sides of the above equation on the left by the row vector $e^t=(0^i1^j0^k)$\footnote{This is the row vector with $i$ $0$'s followed by $j$ $1$'s followed by $k$ $0$'s.} with $i+j+k=N$ gives 
\begin{equation}
\sum_{l=i+1}^{i+j}h(l)=e^t\mathbf{A}\mathbf{x}+\sum_{l=i+1}^{i+j}r(l)
\end{equation}
Since $h(n)=g(n)-l(n)$ and $l(n)\geq 0$ we get the lower bound
\begin{eqnarray}
\sum_{l=i+1}^{i+j}g(l)\geq \sum_{l=i+1}^{i+j}h(l) \geq \sum_{l=i+1}^jr(l)-(-e^t\mathbf{A}\mathbf{x}).
\label{eqn:lb.1}
\end{eqnarray}
Note that equality holds if and only if $l(n)=0$ for $n=i+1,i+2,\ldots,j$. Simplification of the last term on the right hand side of \eqref{eqn:lb.1} gives
\begin{eqnarray}
e^t\mathbf{A}\mathbf{x} = (x(i+1)-\alpha x(i)) + (x(i+2)-\alpha x(i+1)) +\ldots   \nonumber \\
 \ldots +  (x(i+j)-\alpha x(i+j-1)). \nonumber \\
%& \geq &  -\alpha B+\beta,
\label{eqn:lbgen}
\end{eqnarray}
%where $\beta=0$ when $\alpha=1$ or $\Delta P \geq B$, corresponding to the case where the power constraint is redundant.  When the power constraint is non-redundant, i.e. $\Delta P < B$, we get
%\begin{equation}
%\beta=(1-\alpha)s(B-(s+1)\Delta P/2)
%\end{equation}
%with $s=\min\{j,\lfloor B/(\Delta P)\rfloor\}$. Equality holds in \eqref{eqn:lbgen} if and only if $l(n)=0$ for $n=i+1,i+2,\ldots,j$, $x(i)=B$, and $x(i+l)=B-l\Delta P$ for $l=1,2,\ldots,s$.

\subsection{Graph Interpretation and Refinements of the Lower Bound}
\label{sec:theory-b}
For the results that follow, we need the following cumulative quantities:
\begin{eqnarray}
W(n) & = & \sum_{i=1}^nw(i),~n=1,2,\ldots,N, \nonumber \\ 
D(n) & = & \sum_{i=1}^nd(i),~n=1,2,\ldots,N, \nonumber \\
G(n) & = & \sum_{i=1}^n g(i),~n=1,2,\ldots,N. \end{eqnarray}
Also, define
\begin{equation}
R(n)=\sum_{i=1}^n r(i), ~i=1,2,\ldots,N
\label{eqn:cumexdemand}
\end{equation}
 and  
\begin{equation}
R^+(n)=\sum_{i=1}^n \relu{d(i)-w(i)},~i=1,2,\ldots,N.
\end{equation}
Let $\AvExDem=R(N)/(N\Delta)$ be the average excess demand power.

%\subsection{Characterizing the Middle: $\AvGas(B,B/\Delta)$, $B<B^\#$.}
%\label{sec:middle}
%We now consider values of  BESS energy rating $B$ smaller than $B^\#$ defined in the preceding paragraph, assuming $\alpha=1$ (note that the power constraint is inactive since $P=B/\Delta$). 
\begin{figure}[htbp] %  figure placement: here, top, bottom, or page
   \vspace{-5pt}
   \centering
   \includegraphics[width=2.0in]{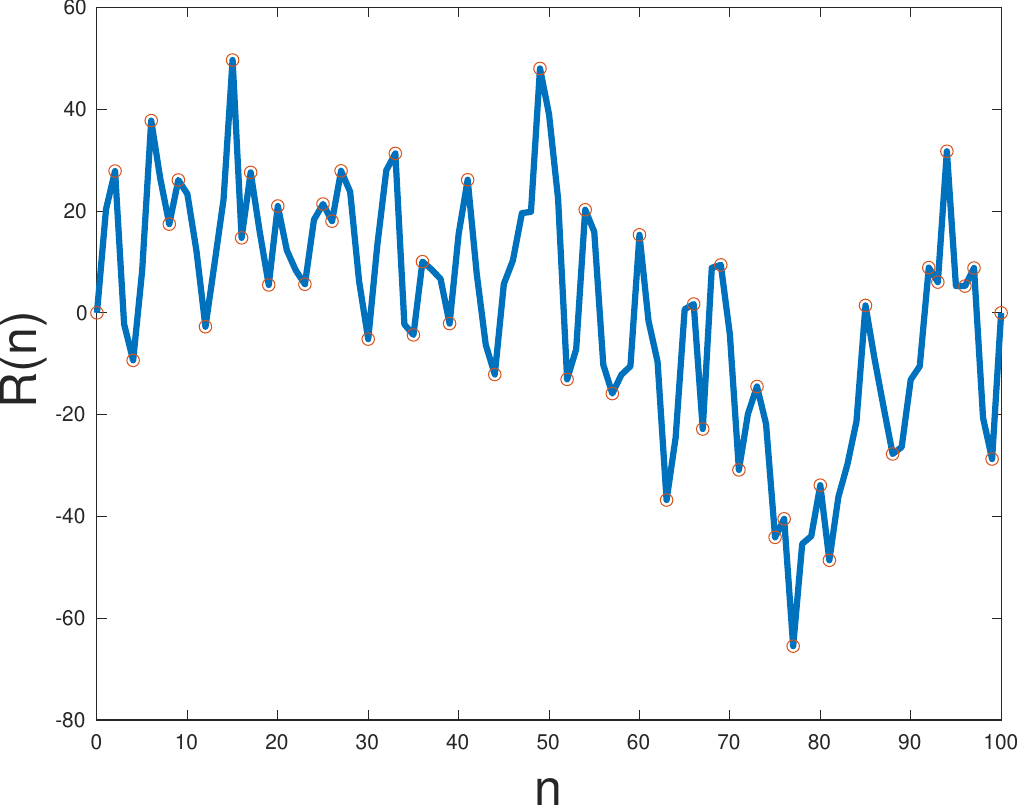} 
   \caption{Random sequence $R$ with local extrema marked with red circles. }
   \vspace{-5pt}
   \label{fig:extreme}
\end{figure}
%From \eqref{eqn:newcompactevo} it follows that 
%\begin{eqnarray}
%x(n) & = & \alpha^{m}x(n-m)+ \nonumber \\
%&  & +\sum_{i=0}^{m-1} \left(-r(n-i) +g(n-i)-l(n-i) \right)\alpha^i, \nonumber \\
%& & ~m=0,1,\ldots,n.
%\end{eqnarray}

It turns out that $\AvGas(B,B/\Delta)$ is  determined by the runs of the same sign in the sequence $r(n)$ or equivalently by the monotone subsequences of $R=(R(0),R(1),\ldots,R(N))$, where $R(n)$ is defined in \eqref{eqn:cumexdemand} for $n=1,2,\ldots,N$ and we have defined $R(0)=0$. We regard $R(0)$ and $R(N)$ as local extreme (local maximum or minimum) points, and mark out the intervening extreme points by their indices, using $m_i$ to indicate the position of a local maximum, and $n_i$ to indicate the position of a local minimum. Thus if $R(0)=0$ is larger than the next extreme point and $R(N)$ is larger than its earlier extreme point, we obtain indices $m_0=0 < n_1 < m_1< n_2< m_2< \ldots < n_K<m_K=N$. On the other hand if $R(0)$ is smaller than the next extreme point, and $R(N)$ is larger than its earlier extreme point, we obtain $n_1=0 < m_1 < n_2 < m_2 < \ldots < n_K < m_K=N$.  The two other cases are  $m_0=0, n_K=N$ and $n_1=0,n_K=N$. Note that in all four cases, the position of the local minimum $n_i$ is smaller than $m_i$. Fig.~\ref{fig:extreme} illustrates the extreme points for an example sequence. The subsequence $r(i)$,  $i\in\{n_i+1\ldots,m_i\}$ is called a positive run of the sequence $r$, and the subsequence with indices  $i\in\{m_i+1,\ldots,n_{i+1}\}$ is called a negative run of the sequence $r$.  
\begin{figure}[htbp] %  figure placement: here, top, bottom, or page
   \centering
   \includegraphics[width=2.5in]{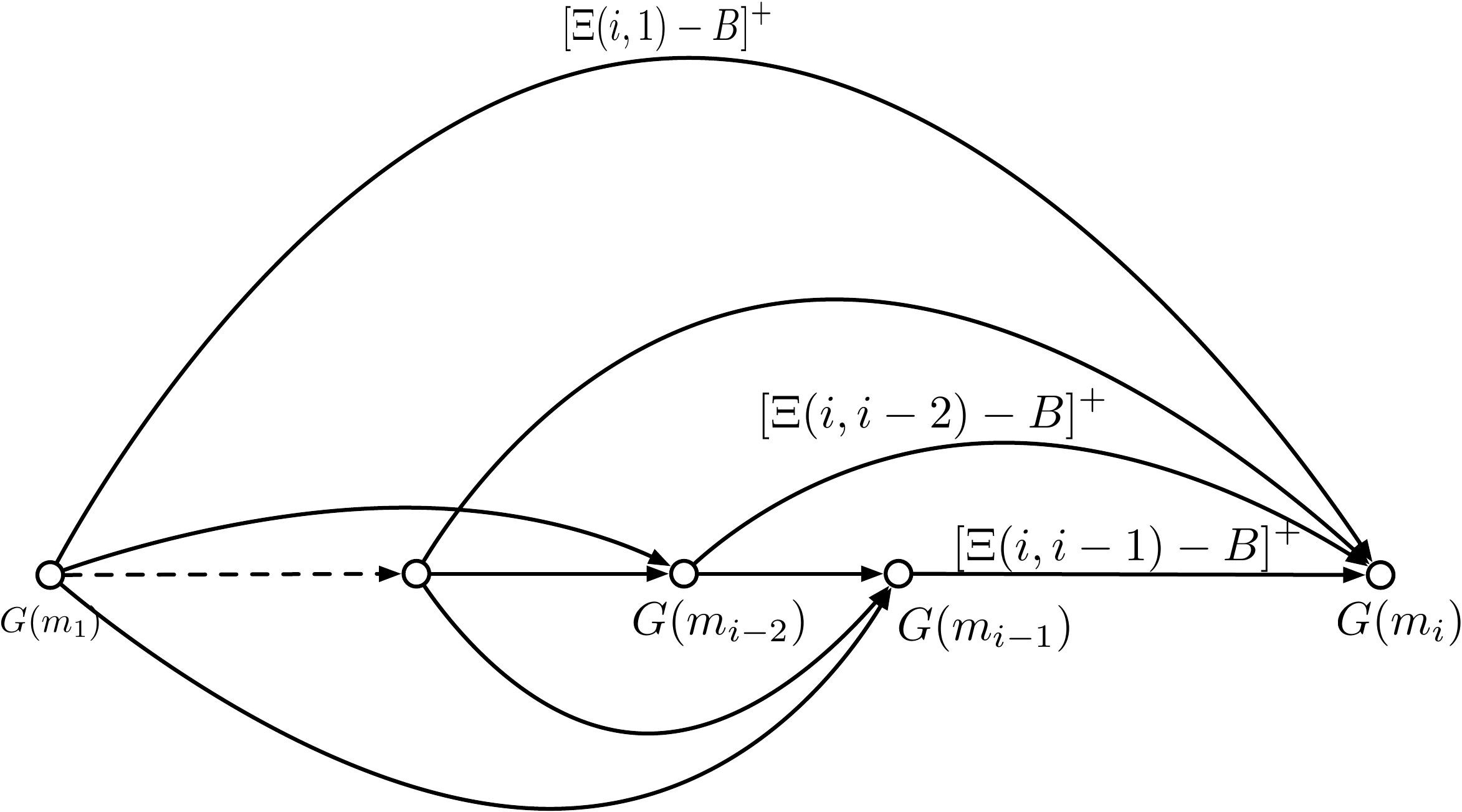} 
   \caption{Graph illustrating computation of $G(N)$. Node $i$ is labeled  $G(m_i)$. An edge originating at  node $i$, presents the sum of the edge label, $\max[\Xi(i,j)-B,0]$, and the label for node $j$, to its destination, node $i$. Node $i$ calculates the maximum value over all incoming edges.}
   \label{fig:compgraph}
\end{figure}
In order to calculate a lower bound for the power alignment function, $\AvGas(B,P)$,  we will work with the cumulative peaker plant energy 
$
G(n)=\sum_{i=1}^n g(i),~n=1,2,\ldots,N.
%\label{eqn:cumpeaker}
$
\begin{figure}[htbp] %  figure placement: here, top, bottom, or page
   \centering
   \includegraphics[width=3in]{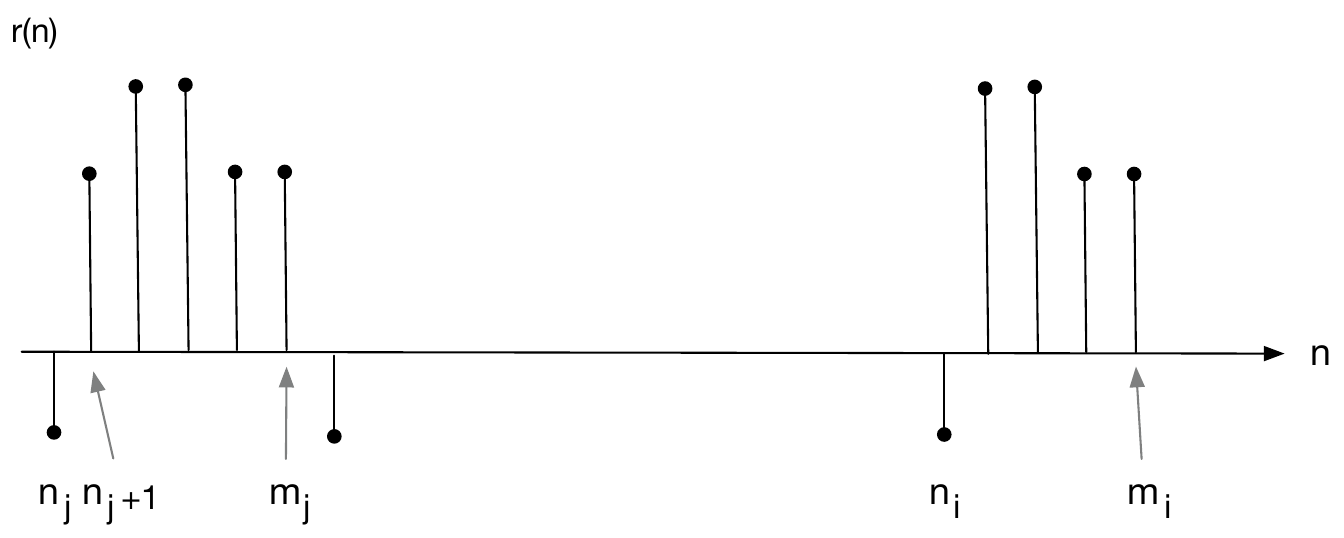} 
   \caption{An example sequences $r(n)$. Location of  local minima, $n_j$'s and local maxima $m_i$'s of corresponding cumulative sequence $R(n)$ are marked. }
   \label{fig:example}
\end{figure}
Consider any local maximum of $R$, say at time-index $m_i$ and consider $G(m_i)$, the total peaker plant energy required up to  time $m_i$. Define 
\begin{equation}
\Xi(i,j)=R(m_i)-R(n_j),~\mbox{for}~j\leq i,
\label{eqn:Xidef}
\end{equation}
which is the total excess demand over wind supply between the extreme points   $n_j$ and $m_i$. Then for any $j \leq i$
\begin{eqnarray}
G(m_i)-G(m_{j-1}) & = & \sum_{l=m_{j-1}+1}^{m_i}g(l) \nonumber \\
&  \geq & \sum_{l=n_j+1}^{m_i} g(l) \nonumber \\
& = & \sum_{l=n_j+1}^{m_i} r(l)-(-e^t\mathbf{A}\mathbf{x})
\end{eqnarray}
where $\mathbf{e}^t=(0^{n_j}1^{m_i-n_j}0^{N-m_i})$. From \eqref{eqn:lbgen} it follows that
\begin{eqnarray}
e^t\mathbf{A}\mathbf{x} & = &  (x(n_j+1)-\alpha x(n_j)) + \nonumber \\
& & + (x(n_j+2)-\alpha x(n_j+1)) +\ldots   \nonumber \\
& &  \ldots +  (x(m_i)-\alpha x(m_i-1)).
\end{eqnarray}

%We now consider different versions of the optimization problem for $\mathbf{e}^t=(0^{n_j}1^{m_i-n_j}0^{N-m_i})$. 

%\subsubsection{$\alpha=1$, and the power constraint is inactive}
%\label{sec:nlpinac}
We now consider the special case, $\alpha=1$, and the power constraint is inactive.
In this case $\mathbf{e}^t\mathbf{A}\mathbf{x}=x(m_i)-x(n_j)\geq 0-B$ and we obtain the bound
\begin{eqnarray}
G(m_i)-G(m_{j-1}) \geq  \sum_{l=n_j+1}^{m_i} r(l)-B.
\end{eqnarray}
Since $G(m_i)-G(m_{j-1}) \geq 0$, this can be strengthened to obtain
\begin{eqnarray}
G(m_i)-G(m_{j-1}) \geq  \relu{\sum_{l=n_j+1}^{m_i} r(l)-B}.
\end{eqnarray}
Based on this we obtain the lower bound on $G$:
\begin{equation}
G(m_i)=\max_{j<i} \left\{ G(m_j)+\relu{\Xi(i,j)-B}\right\}.
\end{equation}
This lower bound on the peaker plant power is shown to be tight (i.e. it coincides with the true value) in App.~\ref{sec:app1}.

%\subsubsection{$\alpha \leq 1$ and the power constraint is inactive}
%Over any positive run of the sequence $r$, say $(r(n_j+1),r(n_j+2),\ldots,r(m_i))$, define the storage credit for $C(j,i)=-e^t\mathbf{A}\mathbf{x}$. It is shown in App.~\ref{sec:app2} that
%\begin{equation}
%C(i,j)=-e^t\mathbf{A}\mathbf{x}\geq r(n_j+1)+r(n_j+2)+\ldots+r(n_j+L)
%\end{equation}
%where $L\leq m_i-n_j$ is chosen such that 
%\begin{eqnarray}
%\sum_{l=1}^{L} r(a+l)/\alpha^l \leq x(a) < \sum_{l=1}^{L+1}r(a+l)/\alpha^l.
%\end{eqnarray}
%Since $G(m_i)-G(m_{j-1}) \geq 0$ we obtain
%\begin{equation}
%G(m_i) \geq G(m_{j-1}) + [\Xi(i,j)-\alpha B+\beta]^+,
%\end{equation}
%where $[x]^+$ is the larger of $x$ and $0$. Note that in the above expression $\beta$ depends on $m_i$ and $n_j$, not $m_i$ and $m_{j-1}$.

We obtain a lower bound for $\PowAl(B,P)$ by recursively computing
\begin{equation}
G(m_i)=\max\{[\Xi(i,j)-\alpha B + \beta]^+ +G(m_{j-1})\}
\label{eqn:condition}
\end{equation}
where the maximum is evaluated over all $j \leq i$.
The computation is illustrated in the diagram shown in Fig.~\ref{fig:compgraph}. Clearly, the recursion \eqref{eqn:condition} can be  solved using dynamic programming. {However, faster methods, based on the separating out negative runs in $r$ which exceed $B$ in absolute sum, can be readily found.}

Equation~\eqref{eqn:condition} and Fig.~\ref{fig:compgraph} explain why the incremental capacity decreases with $B$. Since only positive runs in $r(n)$ which are larger than $B$ contribute to the peaker plant power $\AvGas(B,B/\Delta)$, the incremental capacity  $-d\AvGas/dB$ is proportional  to the number of edges that lie on  the maximum cost path from node $m_0=0$ to $m_K=N$, the last local maximum in the sequence $R$.  As $B$ is increased, an edge will disappear from the maximum cost path because $\Xi(i,j)-B$ is no longer positive. This leads to a reduction in the number of edges on the maximal path from $m_0$ to $m_K$ and a decrease in the incremental capacity. {In other words, we have argued that the power alignment function is convex}.

This theory also reveals that any statistical model that is to be used for estimating BESS capacity must capture this run structure faithfully.
\subsection{Behavior of the Greedy Charging Protocol}
We study the behavior of the greedy charging protocol for runs in the excess demand sequence $\{r(n)\}$. We expect that when there is a positive run between time instants $a$ and $b$  i.e. $r(l)\geq 0$ for $a < l \leq b$, the loss is zero, i.e.   $l(n)=0$ for $a< n \leq b$. On the other hand, when there is a negative run $r(l) < 0$, $a < l \leq b$, the gas is zero, i.e.  $g(n)=0$ for $a < l \leq b$. This is immediately clear from \eqref{eqn:gas} and \eqref{eqn:loss}.
\subsection{Endpoints of the Power Alignment Function}
The following are some theoretical results. (i) The value of $\AvGas(0,0)$ is obtained in  \eqref{eqn:g00} in terms of the wind and demand power time series. This is the average peaker plant power required when there is no energy storage. (ii) The value of $B^\#$, the battery energy rating required to reduce the peaker plant power to $0$ is calculated in \eqref{eqn:necc}, when there is no energy loss in the battery ($\alpha=1$) and when the battery power constraint is inactive ($\Delta P \geq B$). (iii) The key characteristic of the wind and demand power series that characterizes the behavior of $\AvGas(B,P)$ is its run structure, as seen in Sec.~\ref{sec:theory-b} (inactive power constraint and $\alpha=1$). (iv) The minimum peaker plant power obtained by the LP solution coincides with that obtained by the greedy battery charging solution given by \eqref{eqn:recursion1} and \eqref{eqn:gas}, even though the battery states obtained by the two solutions may not coincide. 
%Space constraints preclude us from presenting details, which are available in \cite{LongVersion:2024}.

%\noindent
%{\bf Equal Average (EA) Condition}: The averages of $w$ and $d$ over $N$ samples are equal, i.e. $$\sum_{n=1}^N(w(n)-d(n))=0.$$

%We start by obtaining the values of $\bar{G}(0,0,N)$ and $\bar{L}(0,0,N)$.
First we obtain the average peaker power $\AvGas(0,0)$, i.e. the peaker power when the WF has no battery storage.
\begin{theorem} The average and peak peaker power with no battery storage ($B=0$, $P=0$) are given by
 \begin{eqnarray}
 \AvGas(0,0) & = & \frac{R^+(N)}{N\Delta} \\
 \label{eqn:g00}
 \PeakGas(0,0) & = & \frac{\max_{n=1,2,\ldots,N}(d(n)-w(n))}{\Delta}
 \label{eqn:gp00}
 \end{eqnarray}
    \label{thm:1}
\end{theorem}
\begin{proof}
This follows from the observation that when $B=0$ and $P=0$, peaker plant power must be used to make up the excess demand and follows directly from \eqref{eqn:gas}, by setting $x(i)=0$, $i=1,2,\ldots,N$. 
\end{proof}
\begin{remark} 
For the case $B=P=0$ it follows from conservation of energy that $\AvLoss=\AvGas-\AvExDem$. Further, if the average demand power equals the average wind power (hereafter the EA Condition), then $\AvLoss=\AvGas$.
\end{remark}

%\begin{remark}
%From the previous remark, $\AvExDem=\AvGas-\AvLoss$
%\end{remark}
%\begin{definition}
%The \emph{zero-storage point} is the point on the tradeoff curve with coordinates $(0,\max_nR^+(n))$ and the \emph{zero-storage energy} is $\max_nR^+(n)$.
%\end{definition}
%We now study the value of $B_{max}$ and $P$ for which $\bar{G}$ is the smallest possible. 

%\begin{definition}
%We will refer to $B^\#$ as the \emph{zero-gas storage} size and to $(B^\#,0)$ as the \emph{zero-gas point}.
%\end{definition}
%When $\alpha=1$, the EA condition holds, and the power constraint is inactive it is possible to calculate the minimum required BESS energy capacity. More precisely,
%we seek  the smallest value of $B$ for which $\AvGas(B,B/\Delta)=0$. 
We now determine the BESS size required so that $\AvGas=0$, when there is no loss in the energy storage system ($\alpha=1$) and the power constraint is inactive ($P=B/\Delta$).
\begin{theorem}  
\label{thm:two} 
Let the EA condition hold and let $\alpha=1$. 
(Necessity) If
$\AvGas(B,B/\Delta)=0$  then
\begin{equation}
B \geq  \max_n R(n)-\min_n R(n).
\label{eqn:necc}
\end{equation}
(Sufficiency) Further, there is a value of $x(0)$ such that if $B =  \max_n R(n)-\min_n R(n)$ then $\AvGas(B,B/\Delta)=0$.
\label{thm:2}
\end{theorem}
\begin{proof}
Since $\AvGas=0$ (and hence $\AvLoss=0$ because of EA), it follows that $f(n)$ given by \eqref{eqn:recursion1} must satisfy $0 \leq f(n) \leq B$ and thus $x(n)=f(n)$, $n=1,2,\ldots$.
This means that $x(n)$ obeys the recursion
\begin{equation}
x(n)=x(n-1)+(w(n)-d(n)),~n=1,2,\ldots,N,
\label{eqn:arecurse}
\end{equation}
and 
\begin{equation}
x(0)-x(n)=-W(n)+D(n)=R(n),~n=1,2,\ldots,N.
\label{eqn:battstateev}
\end{equation}
Thus $x(0)-B \leq R(n) \leq x(0)$ from which it follows that 
$\max_n R(n)-\min_n R(n) \leq B$. This proves necessity.
%
%But $-x(0) \leq x(n)-x(0) \leq B-x(0)$ and it follows that
%\begin{eqnarray}
%\max_n -R(n) &  \leq  & B-x(0) \nonumber \\
%-\min_n -R(n) & \leq & x(0).
%\end{eqnarray}
%Combining these two inequalities proves necessity.

In order to prove sufficiency, we need to exhibit a suitable initial condition $x(0)$ such that $x(n)=x(0)-R(n)$ remains in the range $0 \leq x(n) \leq B$, $n=1,2,\ldots,N$. Set $x(0)=\max_n R(n)$. First, since $\max_n R(n) \geq 0$ and $\min_n R(n) \leq 0$ (i.e., due to the EA), $x(0)$ satisfies $0\leq x(0) \leq B$, so this value of $x(0)$ is feasible. Thus $x(n)=\max_n R(n)-R(n)$ satisfies $0\leq x(n) \leq B$, for $n=1,2,\ldots,N$. This proves sufficiency.
\end{proof}
We define $B^\#=\max_n R(n)-\min_n R(n)$.
\begin{remark} ($\alpha=1$ and the power constraint is inactive)
$B^\#$ is the smallest battery energy capacity for which $\sum_n g(n)=0$ and $\sum_n l(n)=0$. However,  the smallest battery energy capacity, $B_g^\#$,  for which only  $\sum_n g(n)=0$, is given by 
$$B^\#_g=\max_{0\leq m\leq n\leq N} \{R(n)-R(m) \}.$$
The proof of this fact follows by choosing the smallest $B$ such that $[\Xi(i,j)-B]^+=0$ for all $i$ and $j\leq i$ in the graph illustrated in Fig.~\ref{fig:compgraph}. In general $B^\#_g \leq B^\#$. As an example consider $\bm{r}=(1^4(-1)^81^5(-1)^51^6(-1)^2)$. In this case $B^\#=8$ and $B^\#_g=6$.
\end{remark}
%\begin{remark}
 
%Setting the battery capacity $B=B^\#$ will work provided the initial state of the battery is chosen properly. As a simple example if the battery is initially discharged, and the initial demand exceeds the energy supplied by the wind, $\bar{G}$ will be greater than zero, no matter how large $B$ is. This theorem shows that a suitable value of $x(0)$ exists such that equality  in \eqref{eqn:necc} is both necessary and sufficient. However, this value of  $x(0)$ cannot be computed in practice because it requires us to look ahead in time at $w(n)$ and $d(n)$. It is not hard to see that  setting $B=2B^\#$ and $x(0)=B^\#$ allows us to get around the feasibility problem by compromising on the size of the BES, assuming that $\max_n R(n)-\min_n R(n)$ can be estimated.
%\end{remark}

\section{Summary and Conclusions}
\label{sec:summary}
A time-series approach is proposed for defining and estimating the capacity and incremental capacity of a BESS captive to a wind farm. The capacity and incremental capacity are obtained from the power alignment function $\AvGas(B,P)$ defined in this paper. A method for computing the  power alignment function is proposed and estimated based on NYSERDA and NYISO data.  Theoretical results are obtained, and a key characteristic of the excess demand for BESS sizing  is identified. Our ongoing and further work is on completing the theoretical calculations and building appropriate statistical models for wind and demand sequences.

\appendix
\subsection{Achievability of the LP optimal solution by the Greedy Charging Protocol}
\label{sec:app1}
Here we consider the greedy charging protocol \eqref{eqn:funditerate} and show that it achieves the lower bound \eqref{eqn:condition} for $\AvGas(B,P)$ when $\alpha=1$ and $P\Delta=B$, i.e. when there is no loss and the power constraint is inactive. 
Our approach is to show that $G(m_i)=\relu{\Xi(i,j)-B}+G(m_j)$ for some $j \leq i$. %and we use the fact that for the greedy charging protocol $G(n_j)=G(m_{j-1})$ which is to say that the greedy charging protocol does not use any peaker plant power on an n-run of $R$. 
Towards this end, observe that the battery state $x(n)$ must equal $B$ for some $n=n_j< m_i$. 
%If not we can always reduce $B$ without affecting the iterations of the greedy charging protocol.
If not, we can always increase the initial state $x(0)$ to  satisfy this condition. 
Let $n_j$ be the largest  index for which  $x(n_j)=B$ and $j \leq i$. Then from the battery state evolution equations~\eqref{eqn:funditerate}--\eqref{eqn:loss} we get
\begin{eqnarray}
x(m_i) & = & x(n_j)-\sum_{i=n_j+1}^{m_i}r(l)+ \sum_{i=n_j+1}^{m_i}h(l)\nonumber \\
& \stackrel{(a)}{=} & B-\Xi(i,j)+\sum_{l=n_j+1}^{m_i}g(l),
\label{eqn:initialfinal}
\end{eqnarray}
where $(a)$ holds because $l(m)=0$, $m=n_j+1,\ldots,m_i$, since, by assumption, the battery state did not reach $B$ at any time index between $n_j+1$ and $m_i$. If $g(m_i)>0$, then $x(m_i)=0$ and 
$\sum_{l=n_j+1}^{m_i}g(l)=\Xi(i,j)-B>0$. Thus $G(m_i)=\relu{\Xi(i,j)-B} +G(n_j)=\relu{\Xi(i,j)-B}+G(m_{j-1})$. If $g(m_i)=0$, then two possibilities exist: \emph{either} $g(l)=0$ for all $l$ between $n_j$ and $m_i$, in which case $\sum_{l=n_j+1}^{m_i}g(l)=0$. From \eqref{eqn:initialfinal} it follows that $\Xi(i,j)-B=-x(m_i) \leq 0$, and thus $G(m_i)=\relu{\Xi(i,j)-B} +G(m_{j-1})$ \emph{or} or there is a largest $m_k$, $n_j< m_k < m_i$ such that $g(m_k)>0$ which means $\sum_{l=m_k+1}^{m_i}g(l)=0$. Thus  $\sum_{l=n_j+1}^{m_i}g(l)=\sum_{l=n_j+1}^{m_k}g(l)$ which means $G(m_i)=\relu{\Xi(i,k)-B} +G(m_k)$ because $\Xi(i,k)-B < 0$.

%\subsection{Lower Bound for $\alpha < 1$, Power constraint inactive}
%\label{sec:app2}
%As seen earlier, for $\mathbf{e}^t=(0^{n_j}1^{m_i-n_j}0^{N-m_i})$, 
%\begin{eqnarray}
%\lefteqn{G(m_i)-G(m_{j-1}) =} & &\nonumber \\
%& &\sum_{l=n_j+1}^{m_i} r(l)-(-e^t\mathbf{A}\mathbf{x})
%\end{eqnarray}
%with
%\begin{eqnarray}
%e^t\mathbf{A}\mathbf{x} & = &  (x(n_j+1)-\alpha x(n_j)) + \nonumber \\
%& & + (x(n_j+2)-\alpha x(n_j+1)) +\ldots   \nonumber \\
%& &  \ldots +  (x(m_i)-\alpha x(m_i-1)).
%\end{eqnarray}

%\vspace{-5pt}
\bibliographystyle{IEEEtran}
\bibliography{reference}
\end{document}